\documentclass[letterpaper, 10 pt, conference]{ieeeconf}

    \usepackage{fancyhdr}
\fancypagestyle{firststyle}
{
    \fancyfoot{} 
   \fancyfoot[L]{\textit{2024 American Control Conference (ACC)}}
}
     \IEEEoverridecommandlockouts 
    \makeatletter

    \let\proof\@undefined
    \let\endproof\@undefined
    \makeatother
    \let\labelindent\relax

    \usepackage{jhoagg}
    \usepackage{mathtools}
    \usepackage{amsfonts}
    \usepackage{amssymb,latexsym}
    \usepackage[psamsfonts]{eucal}
    \usepackage{amsthm}
    \usepackage{graphicx}
    \usepackage[inline]{enumitem}
    \usepackage{indentfirst}
    \usepackage{setspace}
    \usepackage{microtype}
    \usepackage{threeparttable}
    \usepackage{float}
    \usepackage{cleveref}
    \usepackage{afterpage}
    \usepackage{placeins}
    \usepackage{cancel}
    \usepackage{cite}
    \usepackage{leftidx}
    \usepackage{breqn}
    \usepackage[export]{adjustbox}
    \usepackage{comment}
    \usepackage[ruled, linesnumbered, lined]{algorithm2e}
    \usepackage[dvipsnames]{xcolor}
    \usepackage{xcolor}
    \usepackage{soul}
    \usepackage{siunitx}

    \usepackage{tikz}
    \usepackage{pgf}
    \usepackage{pgfplots}
    \usepgflibrary{plothandlers}
    \usepgfplotslibrary{external} 
    \pgfkeys{/pgf/number format/.cd,1000 sep={}}  
    \usetikzlibrary{calc,shapes.misc,plotmarks}
    \pgfplotsset{compat=1.13}
    
    \let\originalleft\left
    \let\originalright\right
    \renewcommand{\left}{\mathopen{}\mathclose\bgroup\originalleft}
    \renewcommand{\right}{\aftergroup\egroup\originalright}

    \newcounter{thm} 
    \newtheorem{theorem}[thm]{\indent Theorem}
    
    \newtheorem{assumption}{\indent Assumption}
    
    \newtheorem{proposition}{\indent Proposition}
    
    \newtheorem{lemma}{\indent Lemma}
    
    \newtheorem{corollary}{\indent Corollary}
    
    \newtheorem{definition}{\indent Definition}

    \newtheorem{example}{\indent Example}

    \newtheorem{Simulation}{Simulation}

    \newtheorem{fact}{\indent Fact}
    
    \newtheorem{conjecture}{\indent Conjecture}
    
    \newtheorem{experiment}{\indent Experiment}
    
    \allowdisplaybreaks

    \renewcommand{\theenumi}{{\it (\alph{enumi})}}
    \renewcommand{\labelenumi}{\theenumi}

    \usepackage{cite}
    \usepackage{accents}
    
    \newlength\figureheight 
    \newlength\figurewidth

    \newcommand\at[2]{\left.#1\right|_{#2}}

    \allowdisplaybreaks
    
    \graphicspath{ {Figures/} }
    \DeclareGraphicsExtensions{.png}
    
    \DeclareMathAlphabet{\mathcal}{OMS}{cmsy}{m}{n} 

    \crefname{equation}{}{}
    
    \usepackage[ruled, linesnumbered, lined]{algorithm2e}
    \SetKwInput{KwInit}{Initialize}
    \SetKwFunction{SafeSet}{SafeSet}
    \SetKwFunction{GetSafe}{GetSafe}
    \SetKwFunction{GetUnsafe}{GetUnsafe}
    
    \SetCommentSty{mycommfont}

\begin{document}

\title{Time-Varying Soft-Maximum Control Barrier Functions for \\Safety in an \textit{A Priori} Unknown Environment}

\author{Amirsaeid Safari and Jesse B. Hoagg
\thanks{A. Safari and J. B. Hoagg are with the Department of Mechanical and Aerospace Engineering, University of Kentucky, Lexington, KY, USA. (e-mail: amirsaeid.safari@uky.edu, jesse.hoagg@uky.edu).}
\thanks{This work is supported in part by the National Science Foundation (1849213, 1932105) and Air Force Office of Scientific Research (FA9550-20-1-0028).}
}
\maketitle
\thispagestyle{firststyle}

\begin{abstract}
This paper presents a time-varying soft-maximum composite control barrier function (CBF) that can be used to ensure safety in an \textit{a priori} unknown environment, where local perception information regarding the safe set is periodically obtained.
We consider the scenario where the periodically obtained perception feedback can be used to construct a local CBF that models a local subset of the unknown safe set. 
Then, we use a novel smooth time-varying soft-maximum function to compose the $N$ most recently obtained local CBFs into a single CBF. 
This composite CBF models an approximate union of the $N$ most recently obtained local subsets of the safe set. 
Notably, this composite CBF can have arbitrary relative degree $r$.
Next, this composite CBF is used as a $r$th-order CBF constraint in a real-time optimization to determine a control that minimizes a quadratic cost while guaranteeing that the state stays in a time-varying subset of the unknown safe set. 
We also present an application of the time-varying soft-maximum composite CBF method to a nonholonomic ground robot with nonnegligible inertia. 
In this application, we present a simple approach to generate the local CBFs from the periodically obtained perception data.
\end{abstract}




\section{Introduction}

Safe autonomous robotic navigation in an \textit{a priori} unmapped environment has application in a variety of domains including search and rescue~\cite{hudson2021heterogeneous}, environmental monitoring~\cite{kress2009temporal}, and transportation~\cite{schwarting2018planning}. 
A techniques for safe navigation include potential field methods\cite{tang2010novel,kirven2021autonomous}, collision cones \cite{sunkara2019collision}, reachable sets \cite{chen2018hamilton}, and barrier function approaches~\cite{prajna2007framework,panagou2015distributed,tee2009barrier,jin2018adaptive,ames2014control,ames2016control,ames2019control}.

Control barrier functions (CBFs) provide a set-theoretic method to obtain forward invariance (e.g., safety) with respect to a specified safe set \cite{wieland2007constructive}.
CBFs can be implemented as constraints in real-time optimization-based control methods (e.g., quadratic programs) in order to guarantee forward invariance and thus, safety \cite{ames2019control}. 
A variety of extensions have been developed recently. 
For example, high-relative degree CBF methods (e.g., \cite{hsu2015control, nguyen2016exponential,xiao2021high}); CBFs for discrete time \cite{zeng2021safety}; and CBFs with time variation or adaptation (e.g., \cite{lindemann2018control,taylor2020adaptive}).

Traditionally, CBFs are assumed to be given offline, that is, constructed offline using \textit{a priori} known information regarding the desired safe set.
However, in situations where the environment is unknown or changing, online construction of valid CBFs could enable safe navigation. 
In this scenario, the objective is to construct a valid CBF in real time based on the current state of the system (e.g., robot) as well as information gathered from the environment (e.g., perception information). 
For example, \cite{srinivasan2020synthesis} uses a support vector machine classifier to build a barrier function using safe and unsafe samples from LiDAR measurements. 
As another example, \cite{long2021learning} synthesizes a barrier function using a  deep neural network trained with sensor data. 
However, when new sensor data is obtained, the barrier function model must be updated, which typically results in discontinuities that can be problematic for ensuring forward invariance and thus, safety.

Discontinuities that arise when new data is available is only one challenge in online construction of a valid CBF. 
Another challenge is that it can be difficult to synthesize a single valid CBF that models a complex environment. 
Thus, there is interest in composing a single valid CBF from multiple CBFs. 
Approaches to compose CBFs include \cite{srinivasan2018control,glotfelter2017nonsmooth,glotfelter2019hybrid,rabiee2023softmin,rabiee2023softmax}.
For example, \cite{glotfelter2017nonsmooth} uses Boolean compositions, which result in nonsmooth barrier functions that use the minimum and maximum functions. 
These nonsmooth barrier functions are required to be relative degree one. 
This nonsmooth barrier function approach is extended in \cite{glotfelter2019hybrid} to allow not only for a nonsmooth composition (e.g., minimum and maximum) but also for time variation, specifically, periodic jumps in the barrier function. 
This extension can be useful for addressing the discontinuities that arise when new information (e.g., perception data) is used to update the barrier function. 
However, this approach is not applicable for relative degree greater than one. 
Thus, \cite{srinivasan2018control, glotfelter2017nonsmooth, glotfelter2019hybrid} cannot be directly applied to higher-relative degree situations such as ground robots with nonnegligible inertia where the safe set is based on position (e.g., avoiding obstacles), or unmanned aerial vehicles with position-based safe sets. 
In contrast to the nonsmooth composition methods, \cite{rabiee2023softmin,rabiee2023softmax} use smooth soft-minimum and soft-maximum functions for composing a single barrier function from multiple barrier functions. 
However, \cite{rabiee2023softmin,rabiee2023softmax} does not address online CBF construction, specifically, updating the CBF based on real-time measurements. 

This paper presents several new contributions. 
First, we present a new time-varying soft-maximum composite CBF construction that can have arbitrary relative degree $r$ with respect to the system dynamics and allows for safety in an \textit{a priori} unmapped environment. 
We consider the scenario where perception feedback information is periodically obtained and used to construct a local CBF, that is, a CBF that models a local subset of the unknown safe set. 
Then, we use a soft-maximum function to compose a single CBF from the $N$ most recently obtained local CBFs. 
Notably, this composition uses not only the soft maximum but also time variation, specifically, a homotopy in time to smoothly move the oldest local CBF out of the soft maximum and newest local CBF into the soft maximum. 
In fact, this homotopy is sufficiently smooth to ensure that the composite soft-maximum CBF is $r$-times continuously differentiable in time and in the state. 
Thus, this composite CBF models an approximate union of the $N$ most recently obtained local subsets of the safe set. 
Next, we use this composite CBF as a $r$th-order CBF constraint in a real-time-optimization control that aims to find an instantaneously optimal control while guaranteeing that the state stays in a subset of the unknown safe set.

We also present an application of the time-varying soft-maximum composite CBF method to a nonholonomic ground robot with nonnegligible inertia.
This robot is equipped with sensing capability (e.g., LiDAR) that periodically detects $P$ points on objects (i.e., obstacles) that are near the robot. 
The robot has the objective of safely navigating the unknown environment. 
We present a simple approach to generate local CBFs from the perception data (e.g., LiDAR points). 
Specifically, we construct an ellipsoidal barrier function for each of the $P$ detected point, where the semi-major axis of the ellipsoid stretches from the detected point to the range of the sensor and the $0$-superlevel set models that area outside the ellipsoid. 
Thus, the $0$-sublevel set of each ellipsoid is a region that we want to avoid because it is not \textit{yet} known to be safe. 
We use a soft-minimum function to compose the $P$ ellipsoidal barrier function and thus, model a local set that is known to be a subset of the safe set from the perception data. 
We then use the $N$ most recently obtained composite soft-minimum CBFs (each one models a local subset of the unknown safe set) in the time-varying soft-maximum barrier function method.

\section{Notation}

Let $\zeta:\BBR^n \to \BBR$ be continuously differentiable. The Lie derivative of $\zeta$ along the vector fields of $\nu:\mathbb{R}^n \to \mathbb{R}^{l}$ is define as 
\begin{equation*}
L_\nu \zeta(x) \triangleq \frac{\partial \zeta(x)}{\partial x}\nu(x).
\end{equation*}

Let $\kappa>0$, and consider the functions $\mbox{softmin}_k : \mathbb{R}^N \to \mathbb{R}$ and $\mbox{softmax}_k:\mathbb{R}^N \to \mathbb{R}$ defined by
\begin{gather}
\mbox{softmin}_\kappa (z_1,\cdots,z_N) \triangleq -\frac{1}{\kappa}\log\sum_{i=1}^Ne^{-\kappa z_i},\label{eq:softmin}\\
\mbox{softmax}_\kappa (z_1,\cdots,z_N) \triangleq \frac{1}{\kappa}\log\sum_{i=1}^Ne^{\kappa z_i} - \frac{\log N }{\kappa},\label{eq:softmax}
\end{gather}
which are the soft minimum and soft maximum, respectively.

The next result relates the soft minimum to the minimum and the soft maximum to the maximum.
\begin{fact} \label{fact:softmin_limit}
\rm{
Let $z_1,\cdots, z_N \in \mathbb{R}$. 
Then,
\begin{align*}
  \min \, \{z_1,\cdots,z_N\} - \frac{\log N }{\kappa} 
    &\le \mbox{softmin}_\kappa(z_1,\cdots,z_N)\\
    &< \min \, \{z_1,\cdots,z_N\},
\end{align*}
and
\begin{align*}
    \max\,\{z_1,\cdots,z_N\}  - \frac{\log N}{\kappa}
 &< \mbox{softmax}_\kappa(z_1,\cdots,z_N) \\
 &\le \max \, \{z_1,\cdots,z_N\}.
\end{align*}
}
\end{fact}
Fact~\ref{fact:softmin_limit} shows that as $\kappa \to \infty$, $\mbox{softmin}_\kappa$ converges to the minimum and $\mbox{softmax}_\kappa$ converges to the maximum.

\section{Problem Formulation}\label{sec:problem formulation}

Consider
\begin{equation}\label{eq:affine control}
\dot x(t) = f(x(t))+g(x(t)) u(t), 
\end{equation}
where $x(t) \in \mathbb{R}^n $ is the state, $u(t) \in \mathbb{R}^m$ is the input, $x(0) = x_0 \in \mathbb{R}^n$ is the initial condition, and $f: \mathbb{R}^n \to \mathbb{R}^n$ and $g: \mathbb{R}^n \to \mathbb{R}^{n \times m}$ are locally Lipschitz continuous. 
Let $h_{\rm{s}}: \mathbb{R}^n \to \mathbb{R}$, and define the \textit{safe set} 
\begin{equation}
   S_{\rm{s}} \triangleq \{x \in \mathbb{R}^n \colon h_{\rm{s}}(x) \geq 0 \}, 
\end{equation}
which is not assumed to be known. 
Let $T_\rms > 0$, and assume that for all $k \in \mathbb{N}$, we obtain perception feedback at time $kT_\rms$ in the form of a function $b_k : \mathbb{R}^n \to \mathbb{R}$ such that:

\begin{enumerate}[leftmargin=0.9cm]
	\renewcommand{\labelenumi}{(A\arabic{enumi})}
	\renewcommand{\theenumi}{(A\arabic{enumi})}

 \item\label{con1}
$S_{k} \triangleq \{x \in \mathbb{R}^n : b_k(x) \ge 0\} \subseteq S_{\rm{s}}$.
    
\item\label{con2}
If $x(kT_\rms) \in S_{\rm{s}}$, then $x(kT_\rms) \in S_{k}$.

\item\label{con3}
There exist a positive integer $r$ such that $b_k(x)$ is $r$-times continuously differentiable, and for all $x \in \mathbb{R}^n$, $L_gb_k(x)=L_gL_fb_k(x)=\cdots=L_gL_f^{r-2}b_k(x)=0$ and $L_gL_f^{r-1}b_k(x) \neq 0$.

\end{enumerate}
Note that the perception information $b_k$ can be obtained from a variety of CBF synthesis methods (e.g., \cite{srinivasan2020synthesis,long2021learning}).

Next, consider the \textit{desired control} $u_{\rm{d}} : [0,\infty) \to \mathbb{R}^m$. 
The objective is to design a full-state feedback control $u:\mathbb{R}^n \to \mathbb{R}^m$ such that for all $t \ge 0$, $\| u(x(t)) -u_{\rm{d}}(t) \|_2$ is minimized subject to the safety constraint that $x(t) \in S_{\rm{s}}$.

All statements in this paper that involve the subscript $k$ are for all $k \in \mathbb{N} \triangleq \{ 0, 1, 2, \ldots \}$.

\section{Time-Varying CBF for Unknown Safe Set} \label{sec:Method}

Let $\eta:\mathbb{R} \to [0,1]$ be $r$-times continuously differentiable such that the following condition hold: 
\begin{enumerate}[leftmargin=0.9cm]
	\renewcommand{\labelenumi}{(C\arabic{enumi})}
	\renewcommand{\theenumi}{(C\arabic{enumi})}

 \item\label{con: con1_g}
For all $t\in (-\infty,0]$, $\eta(t) = 0$.

\item\label{con: con2_g}
For all $t\in [1,\infty)$, $\eta(t) = 1$.

\item\label{con: con3_g}
For all $i \in \{1,\cdots,r\}$, $\at{\frac{{\rm d}^i \eta(t)}{{\rm d}t^i}}{t=0} = 0$.

\item\label{con: con4_g}
For all $i \in \{1,\cdots,r\}$, $\at{\frac{{\rm d}^i \eta(t)}{{\rm d}t^i}}{t=1} = 0$. 
\end{enumerate}

The following example provides one possible choice of $\eta$ that satisfies \ref{con: con1_g}--\ref{con: con4_g}.

\begin{example}\label{ex:g}\rm
Let $\lambda \ge 1$, and consider $\eta:\mathbb{R} \to [0,1]$, defined by
\begin{equation}\label{eq:smoothstep}
\eta(t) \triangleq
\begin{cases}
        0, & t <0, \\
        \left(\lambda t\right)^{r+1} \sum_{j=0}^{r} \binom{r+j}{j}\binom{2r+1}{r-j}(-\lambda t)^j, & 0 \le t \le \frac{1}{\lambda}, \\
        1, & t>\frac{1}{\lambda}.
\end{cases}
\end{equation}
\Cref{fig:eta} is a plot of $\eta$ given by \Cref{eq:smoothstep} for different value of $\lambda$. 
\end{example}

\begin{figure}[t!]
\center{\includegraphics[width=0.44\textwidth,clip=true,trim= 0.42in 0.25in 1.1in 0.6in] {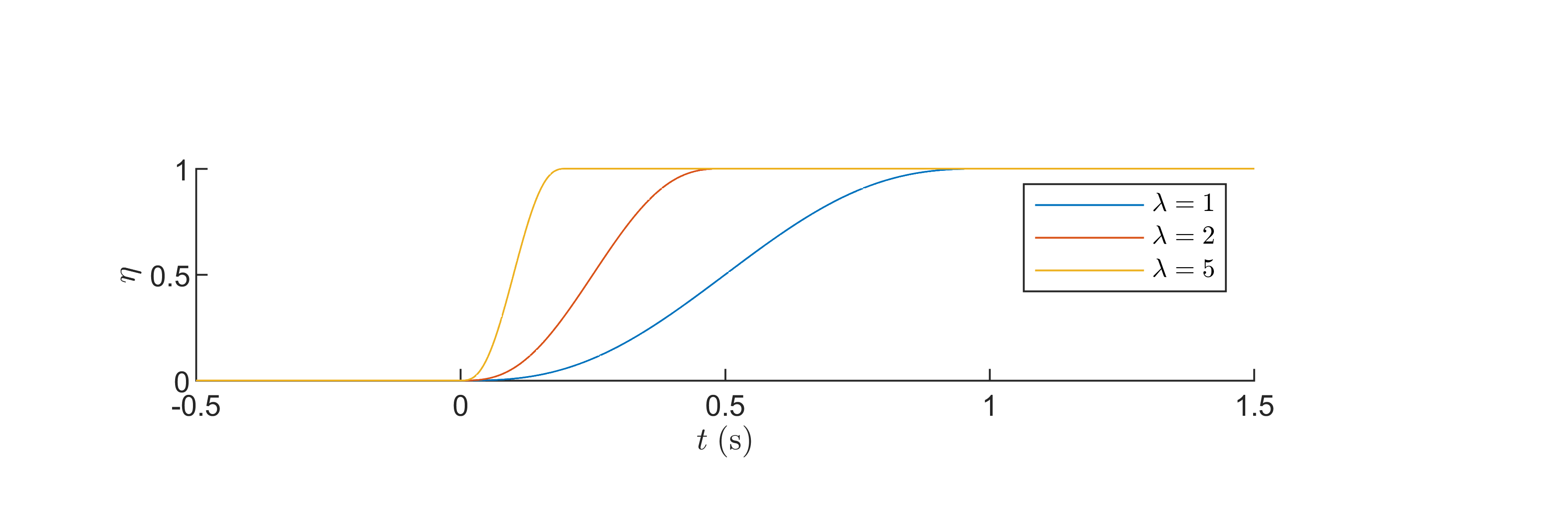}}
\caption{Example of $\eta$.}\label{fig:eta}
\end{figure}

Let $N$ be a positive integer, and consider $h:\mathbb{R}^n \times [0,\infty) \to \mathbb{R}$ such that for all $t \in [kT_\rms,(k+1)T_\rms)$,
\begin{align}\label{eq:softmax h}
h(x,t) &\triangleq \mbox{softmax}_\kappa \Bigg (b_{k-1}(x), \cdots, b_{k-N+1}(x),\nonumber\\
&\quad \eta\left(\frac{t}{T_\rms}-k\right)b_k(x)  + \left(1-\eta\left(\frac{t}{T_\rms}-k\right)\right)b_{k-N}(x) \Bigg ).
\end{align}
Note that $h$ is constructed from the $N$ most recently obtained perception feedback functions $b_k,\ldots,b_{k-N}$. 
The next result demonstrates that $h$ is $r$-times continuously differentiable. T
he proof is omitted for space considerations. 

\begin{fact}\label{fact:rth order continuity}
\rm{ The function $h(x,t)$ is $r$-times continuously differentiable with respect to $x$ and $t$.
}
\end{fact}

Next, we define the $0$-superlevel set of $h$. 
Specifically, consider $S \colon [0, \infty) \to \BBR^n$ defined by
\begin{equation} \label{eq:safe set final}
    S(t) \triangleq \{x \in \mathbb{R}^n \colon h(x,t) \geq 0 \}.
\end{equation}
The next result shows that $S(t)$, which is constructed from real-time perception data, is a subset of the safe set $S_{\rm{s}}$

\begin{fact}\label{fact:S(t)}
\rm{
   For all $t \in [0,\infty), S(t) \subseteq S_{\rm{s}}$. 
   }
\end{fact}


For $i \in \{1, \cdots, r\}$, let $\alpha_i:\mathbb{R} \to \mathbb{R}$ be a extended class $\mathcal{K}$ function that is $(r-i)$-times continuously differentiable. 
Next, consider $\psi_0:\mathbb{R}^n \times [0, \infty) \to \mathbb{R}$ defined by $\psi_0(x,t) \triangleq h(x,t)$, and for $i \in \{1,\cdots,r-1\}$, consider $\psi_i:\mathbb{R}^n \times [0, \infty) \to \mathbb{R}$ defined by
\begin{equation}\label{eq:HOCBF}
        \psi_i(x,t) \triangleq \frac{\partial \psi_{i-1}(x,t)}{\partial t} + L_f \psi_{i-1}(x,t) +\alpha_{i}(\psi_{i-1}(x,t)).
\end{equation}

Next, we design a control that for all $t \ge 0$, seeks to minimize $\| u(x(t)) -u_{\rm{d}}(t) \|_2$ subject to the safety constraint that $x(t) \in S(t) \subset S_\rms$.  
For all $t \in [0,\infty)$ and all $x \in \mathbb{R}^n$, the control is given by
\begin{subequations}\label{eq:qp}
\begin{align}
 & u(x,t) \triangleq \underset{ \hat{u} \in \mathbb{R}^m}{\mbox{argmin}}  \, 
\| \hat{u} - u_{\rm{d}}(t)\|^2 \label{eq:qp_softmin.a}\\
\text{subject to}\nonumber \\
& \begin{aligned}\label{eq:HOCBF_const}
        & \frac{\partial \psi_{r-1}(x,t)}{\partial t} + L_f \psi_{r-1}(x,t)  \\
        &\qquad + L_g\psi_{r-1}(x,t)\hat{u} +\alpha_{r}(\psi_{r-1}(x,t)) \ge 0.
\end{aligned}
\end{align}
\end{subequations}





To analyze the control~\Cref{eq:softmax h,eq:safe set final,eq:HOCBF,eq:qp}, for all $i \in \{1,\cdots,r-1\}$ define
\begin{equation}\label{eq:HOCBF set}
    C_i(t) \triangleq \{x \in \mathbb{R}^n \colon \psi_{i}(x,t) \ge 0\},
\end{equation}
and 
\begin{equation}\label{eq:Common set}
    \Gamma(t) \triangleq S(t) \cap C_1(t) \cap \cdots \cap C_{r-1}(t). 
\end{equation}

The next result is the main result on the control~\Cref{eq:softmax h,eq:safe set final,eq:HOCBF,eq:qp}, which is constructed from real-time perception data $b_k$. 
%
%
This result shows that the control guarantees safety and follows from \cite[Theorem~3]{xiao2021high}.

\begin{theorem}\label{Th:Main th}
\label{thm:softmin}
\rm{ 
Consider \eqref{eq:affine control}, where \ref{con1}--\ref{con3} are satisfied, and consider $u$ given by~\Cref{eq:softmax h,eq:safe set final,eq:HOCBF,eq:qp}. 
Assume that for all $t \in [0,\infty)$ and all $x \in S(t)$, $L_g\psi_{r-1}(x,t) \ne 0$. 
Let $x_0 \in \Gamma(0)$. 
Then, for all $t \in [0,\infty)$, $x(t) \in \Gamma(t) \subseteq S_\rms$.
}
\end{theorem}

\section{Application to a Nonholonomic Ground Robot}

Consider the nonholonomic ground robot modeled by \eqref{eq:affine control}, where
\begin{equation*}
    f(x) = \begin{bmatrix}
     v\cos\theta \\
     v\sin\theta \\
    0 \\
    0
    \end{bmatrix}, 
    \,
    g(x) = \begin{bmatrix}
    0 & 0\\
    0 & 0\\
    1 & 0 \\
    0 & 1
    \end{bmatrix}, 
    \,
    x = \begin{bmatrix}
    q_x\\
    q_y\\
    v\\
    \theta
    \end{bmatrix}, 
    \,
    u = \begin{bmatrix}
    u_1\\
    u_2
    \end{bmatrix}, 
\end{equation*}
and $\matl{cc} q_x & q_y \matr^\rmT$ is the robot's position in an orthogonal coordinate frame, $v$ is the speed, and $\theta$ is the direction of the velocity vector.

Let $q_\rmg\in\mathbb{R}^2$ be the goal location, which is denoted by $q_\rmg = \matl{cc} q_{x,\rmg} & q_{y,\rmg} \matr^\rmT $. 
The desired control is
\begin{equation*}
    u_\rmd\triangleq \begin{bmatrix}
    u_{\rmd_1}\\u_{\rmd_2}
\end{bmatrix}, 
\end{equation*}
 where
\begin{gather*}
u_{\rmd_1} \triangleq  -(k_1+k_3)v + (1+k_1k_3)\rho\cos\delta + k_1(\rho k_2+v)\sin^2\delta,\\
u_{\rmd_2} \triangleq \left ( k_2+\frac{v}{\rho} \right )\sin\delta,\\
\rho\triangleq\sqrt{(q_x-q_{x,\rmg})^2+(q_y-q_{y,\rmg})^2}, \\
\delta\triangleq\mbox{atan2}(q_y-q_{y,\rmg},q_x-q_{x,\rmg})-\theta + \pi,
\end{gather*}
and $k_1 = 0.5$, $k_2 = 3$, and $k_3 = 3$.
This desired control is designed using a process similar to \cite[pp.~30--31]{de2002control}.

Next, define $\chi(x) \triangleq \matl{cc} I_2 & 0_{2\times2} \matr^\rmT x$, which extracts the position elements of the state. 
The robot is equipped with sensing capability (e.g., LiDAR) that can detect $P$ points on objects that are inside an $\bar r > 0$ detection radius of the robot's position. 
At time $k T_\rms$, the robot obtains the detection points $(r_{1,k},\theta_{1,k}),\cdots,(r_{P,k},\theta_{P,k})$, which are the polar coordinates of detected objects relative to the robot's position $\chi\left(x\left(kT_\rms\right)\right)$, where $r_{i,k} \in [0,\bar r]$ and $\theta_{i,k} \in [0,2\pi)$.
If no object is detected at angle $\theta_{i,k}$, then $r_{i,k} = \bar r$.

    


Next, for each detected point, we define a function whose $0$-level set is an ellipse that contains the detected point and extends from the detected point to the edge of the detection radius. 
Specifically, for $i \in \{1,\cdots,P\}$, consider $\sigma_{i,k}:\BBR^4 \to \BBR$ defined by 
\begin{equation}\label{eq:ellipse}
    \sigma_{i,k}(x) \triangleq \left ( \chi(x) - c_{i,k} \right )^\rmT R_{i,k} P_{i,k} R^\rmT_{i,k} \left (\chi(x) - c_{i,k} \right ) - 1, 
\end{equation}
where 
\begin{gather*}
    c_{i,k} \triangleq \chi\left(x\left(kT_\rms\right)\right) + \frac{\bar r+r_{i,k}}{2}\begin{bmatrix}
        \cos\theta_{i,k} \\ \sin\theta_{i,k}
    \end{bmatrix},\\
    R_{i,k} \triangleq \begin{bmatrix}
        \cos\theta_{i,k} & -\sin\theta_{i,k} \\ \sin\theta_{i,k}& \cos\theta_{i,k}
    \end{bmatrix},\\
    P_{i,k} \triangleq \begin{bmatrix}
        a_{i,k}^{-2}&0 \\ 0 & d_w^{-2}
    \end{bmatrix},
\end{gather*}
and $a_{i,k} \triangleq \frac{\bar r-r_{i,k}}{2} + d_s$ is the length of the semi-major axis, $d_w >0$ is the length of the semi-minor axis, and $d_s \ge 0$ is a safety margin. 
Note that the $0$-superlevel set of $\sigma_{i,k}$ is the area outside the ellipse that contains the $i$th detected point.
For the examples in this section, we let $d_w = 0.3$, $d_s = 0.3$, and $\bar r = 5$.

We consider two cases for sensing capability: (1) a $360^\circ$ field of view, and (2) a limited (e.g., $100^\circ$) field of view.

\subsection{$360^\circ$ Field of View}\label{sec:360_FOV}

We now use the ellipses $\sigma_{1,k},\cdots,\sigma_{p,k}$, which are generated from the raw perception data, to construct the perception feedback function $b_k$. 
We use a soft minimum to combine the ellipses and construct the perception feedback function $b_k$. 
Specifically, we let 
\begin{equation} \label{eq:bk_ex1}
    b_k(x) = \mbox{softmin}_{\kappa_1} \left(\sigma_{1,k}(x),\cdots,\sigma_{P,k}(x)\right).
\end{equation}

\Cref{fig:ex1_safeset} shows a map of the \textit{a priori} unknown environment that the ground robot must navigate. 
The figure also shows the $0$-level sets of $\sigma_{1,0},\cdots,\sigma_{P,0}$, which are constructed from the raw perception data at $t=0$~s.
Notice that each ellipse contains a detected point on the obstacles, and the detected point is on the interior of the ellipse because the safety margin $d_s > 0$ which makes the length of the semi-major axis greater than the distance from the detected point to the boundary of the detection radius.

We implement the soft-maximum CBF control~\Cref{eq:softmax h,eq:safe set final,eq:HOCBF,eq:qp} is implemented with $T_\rms = 0.2 s$, $\kappa_1 = 20$, $\kappa = 20$, $N = 2$, $P = 100$, $\alpha_1 =20$, $\alpha_2 =20$, and $\eta$ given by \Cref{ex:g} where $r=2$ and $\lambda = 1$.

\Cref{fig:ex1_map} shows the map and closed-loop trajectories for $x_0 = [\,5\quad 2\quad0\quad0\,]^\rmT$ with 3 different goals locations $q_\rmg = [\,13\quad5\,]^\rmT$, $q_\rmg = [\,10\quad13\,]^\rmT$, and $q_\rmg = [\,1\quad11\,]^\rmT$. 
In all cases, the robot position converges to the goal location while satisfying safety constraints. 
\Cref{fig:ex1_h,fig:ex1_state} provide time histories of relevant signals for the case where $q_\rmg = [\,13\quad5\,]^\rmT$. \Cref{fig:ex1_h} shows $h$ and $\psi_1$, which are always positive, indicating that the safety constraint is satisfied. 
\Cref{fig:ex1_state} shows $q_x$, $q_y$, $v$, $\theta$, $u_\rmd$, and $u$. 
Note that the control \Cref{eq:softmax h,eq:safe set final,eq:HOCBF,eq:qp} modifies the desired control signals to avoid collisions with detected obstacles.

\begin{figure}[t!]
\center{\includegraphics[width=0.44\textwidth,clip=true,trim= 0.30in 0.15in 0.7in 0.5in] {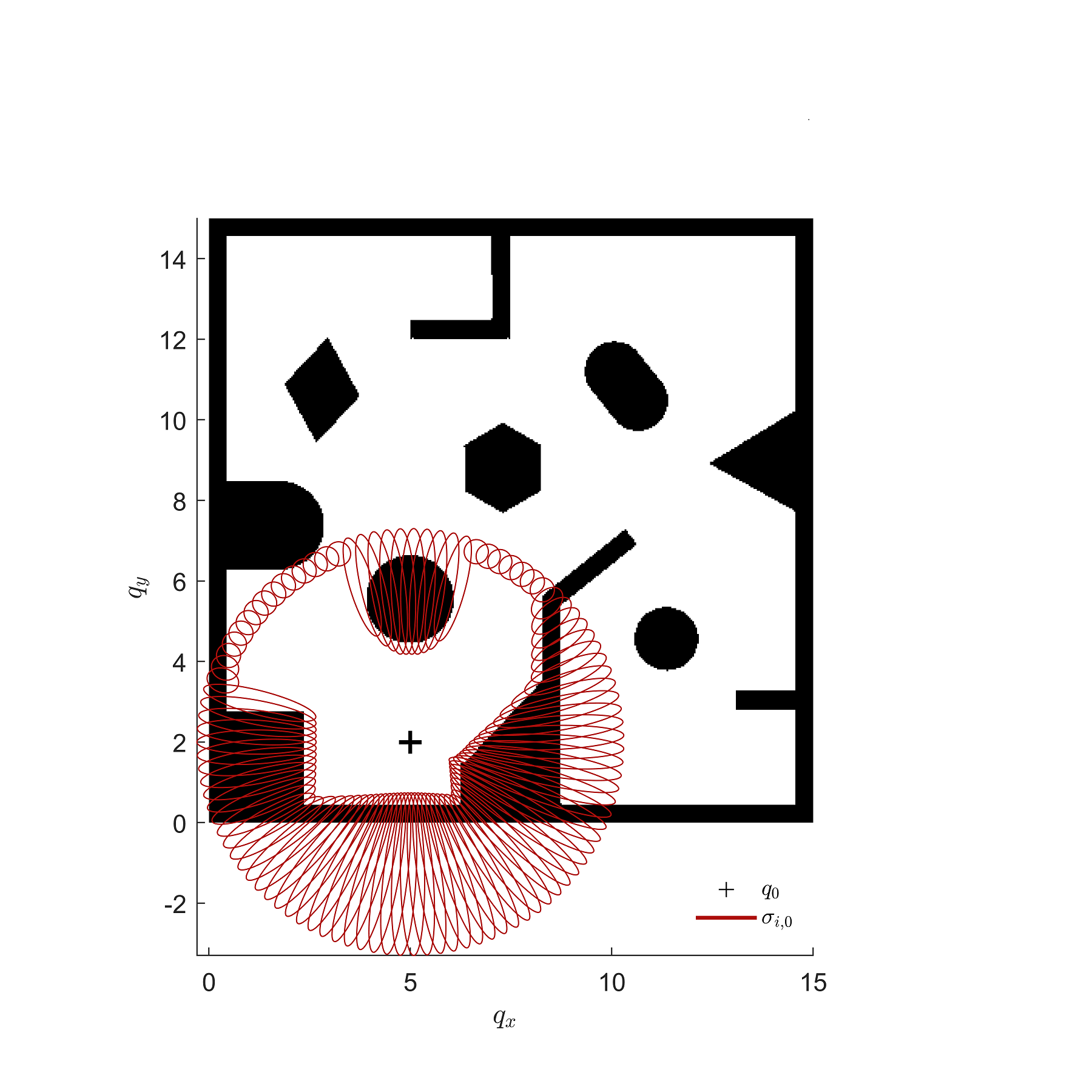}}
\caption{Map of \textit{a priori} unknown environment.}\label{fig:ex1_safeset}
\end{figure} 
\begin{figure}[t!]
\center{\includegraphics[width=0.44\textwidth,clip=true,trim= 0.25in 0.35in 1.1in 1.1in] {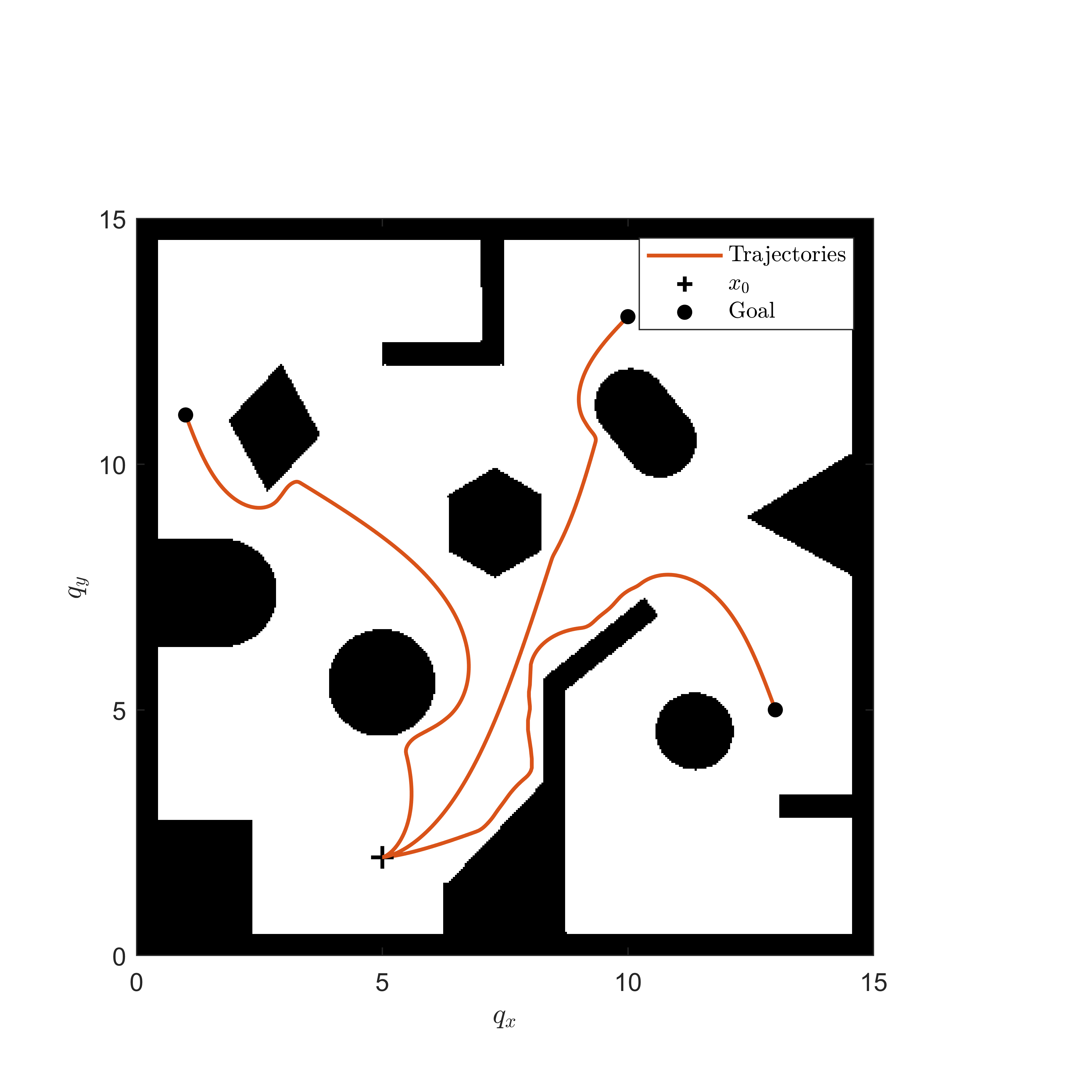}}
\caption{Three closed-loop trajectories with $360^{\circ}$ LiDAR.}\label{fig:ex1_map}
\end{figure} 
\begin{figure}[t!]
\center{\includegraphics[width=0.44\textwidth,clip=true,trim= 0.3in 0.27in 1.1in 0.5in] {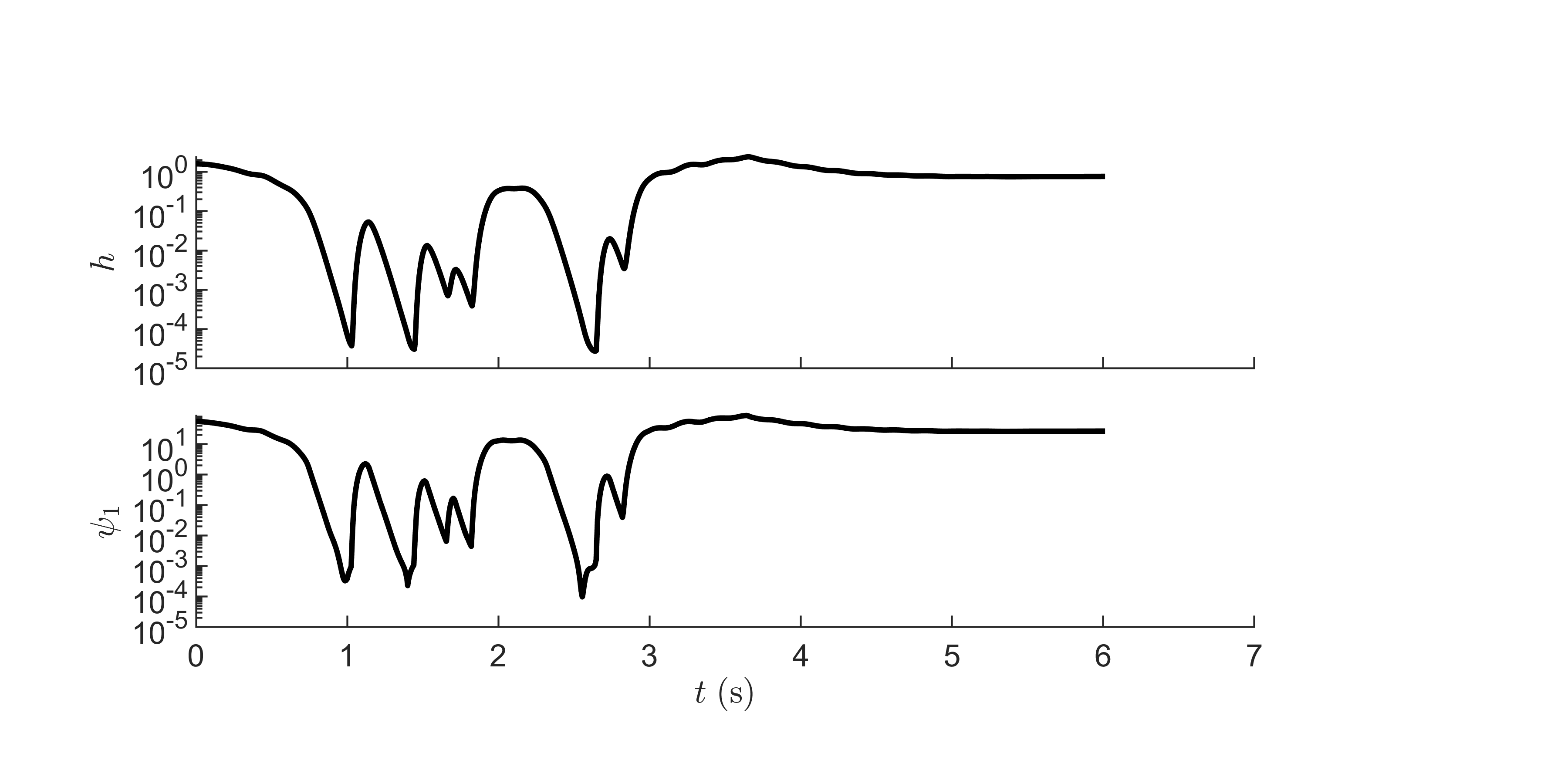}}
\caption{$h$ and $\psi_1$ for $q_\rmg = [\,13\quad5\,]^\rmT$ .}\label{fig:ex1_h}
\end{figure} 
\begin{figure}[t!]
\center{\includegraphics[width=0.44\textwidth,clip=true,trim= 0.4in 0.47in 1.1in 0.7in] {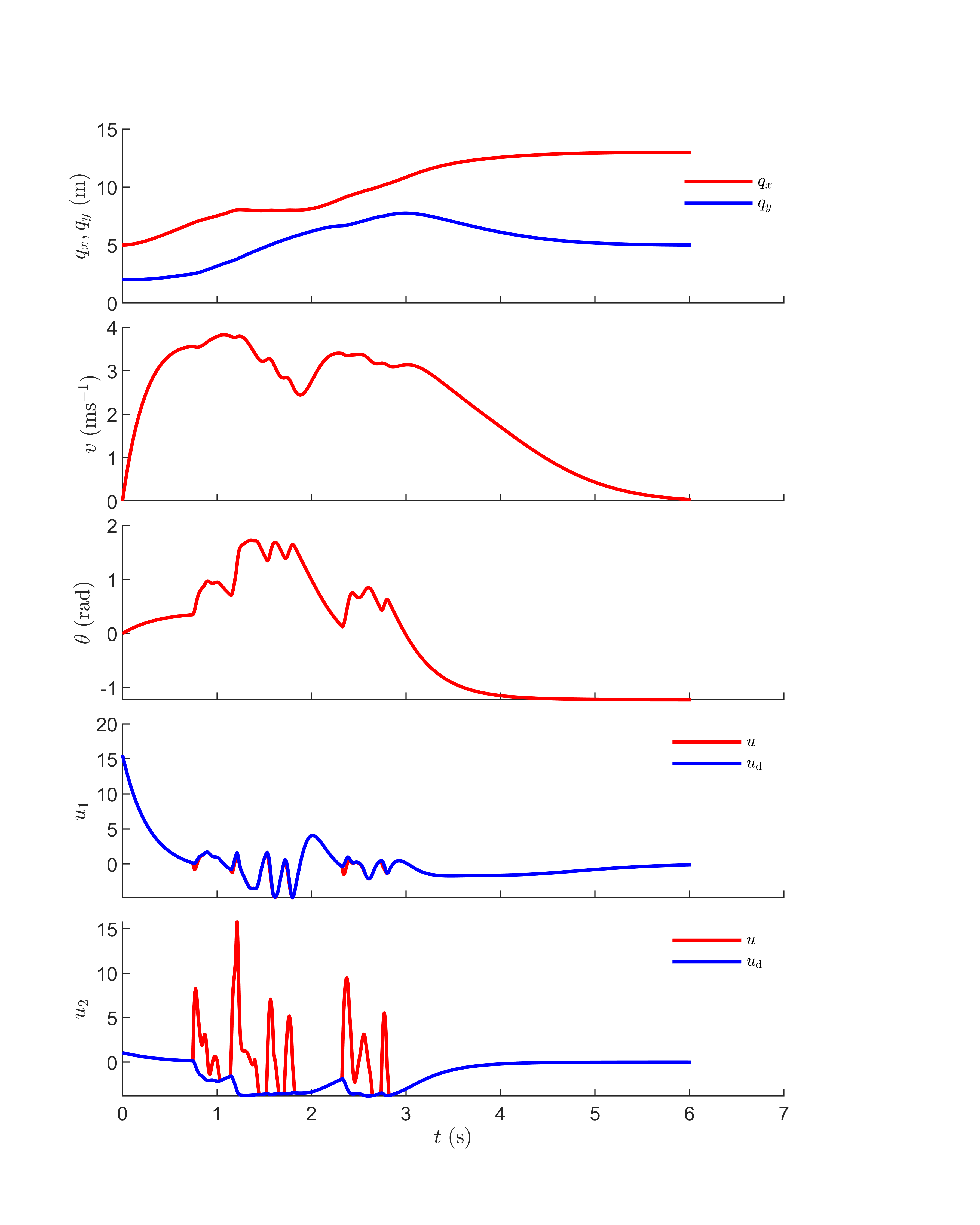}}
\caption{$q_x$, $q_y$, $v$, $\theta$, $u_\rmd$, and $u$ for $q_\rmg = [\,13\quad5\,]^\rmT$.}\label{fig:ex1_state}
\end{figure}

\subsection{Limited Field of View}\label{sec:limited_FOV}

Next, we consider the scenario where the robot's sensing capability has a limited field of view (i.e., $100^{\circ}$) as shown in \Cref{fig:FOV_safeset}. 
To model the limited field of view, we sample $L$ points along the boundary of the field of view at time $kT_\rms$. 
For each $i \in \{1,\cdots,L\}$, we construct an ellipse denoted by $\gamma_{i,k}:\BBR^4 \to \BBR$ using \eqref{eq:ellipse} that contains the point on the boundary and goes to the maximum range $\bar r$.

We now use the ellipses $\sigma_{1,k},\cdots,\sigma_{p,k}$, which are generated from the raw perception data, and ellipses $\gamma_{1,k},\cdots,\gamma_{p,k}$, which generated from the field of view boundary, to construct the perception feedback function $b_k$. 
We use a soft minimum to combine the ellipses and construct the perception feedback function $b_k$. 
Specifically, we let  
\begin{gather*}
     b_k(x) = \mbox{softmin}_{\kappa_1} (\sigma_{1,k}(x),\cdots,\sigma_{P,k}(x), \gamma_{1,k}(x),\cdots,\gamma_{L,k}(x) ).
\end{gather*}

\Cref{fig:ex1_FOV_safeset} shows a map of the \textit{a priori} unknown environment that the ground robot must navigate. 
The figure also shows the $0$-level sets of $\sigma_{1,0},\cdots,\sigma_{P,0}$ and $\gamma_{1,0},\cdots,\gamma_{L,0}$, which are constructed from the raw perception data and boundary of the $100^\circ$ field of view obtained at $t=0$ s.


The soft-maximum CBF control \Cref{eq:softmax h,eq:safe set final,eq:HOCBF,eq:qp} is implemented with $T_\rms = 0.2 s$, $\kappa_1 = 30$, $\kappa = 30$, $N = 6$, $P = 100$, $L = 400$, $\alpha_1 =30$,  $\alpha_2 =30$, and $\eta$ given by \Cref{ex:g}, where $r = 2$ and $\lambda = 1$.

\begin{figure}[t!]
\center{\includegraphics[scale = 0.4,clip=true,trim= 1.2in 0.4in 0.7in 0.3in] {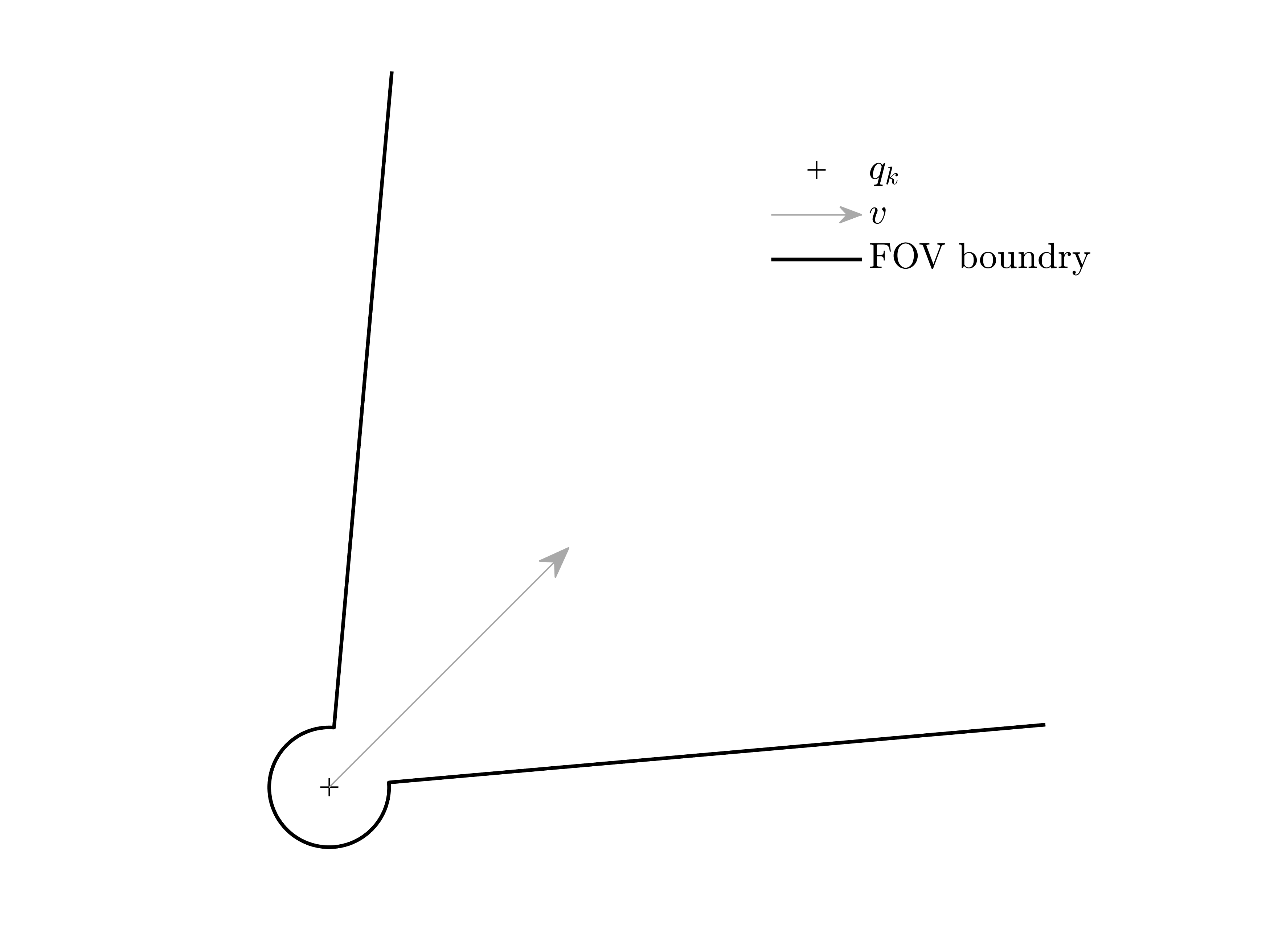}}
\caption{Boundry of the $100^\circ$ field of view at $t = kT_\rms$}\label{fig:FOV_safeset}
\end{figure}

\Cref{fig:ex1_FOV_map} shows the map and closed-loop trajectories  for $x_0 = [\,5\quad2\quad0\quad0\,]^\rmT$ with 3 different goals locations $q_\rmg = [\,11\quad2\,]^\rmT$, $q_\rmg = [\,10\quad13\,]^\rmT$, and $q_\rmg = [\,1\quad11\,]^\rmT$.
In all cases, the robot converges to the goal location while satisfying safety. 
For the case where $q_\rmg = [\,13\quad5\,]^\rmT$, \Cref{fig:ex1_FOV_h,fig:ex1_FOV_state} provide time histories of relevant signals. 
\Cref{fig:ex1_FOV_h} shows that $h$ and $\psi_1$ are always positive, and \Cref{fig:ex1_FOV_state} shows $q_x$, $q_y$, $v$, $\theta$, $u_\rmd$, and $u$.

\begin{figure}[t!]
\center{\includegraphics[width=0.44\textwidth,clip=true,trim= 0.30in 0.15in 0.7in 0.5in] {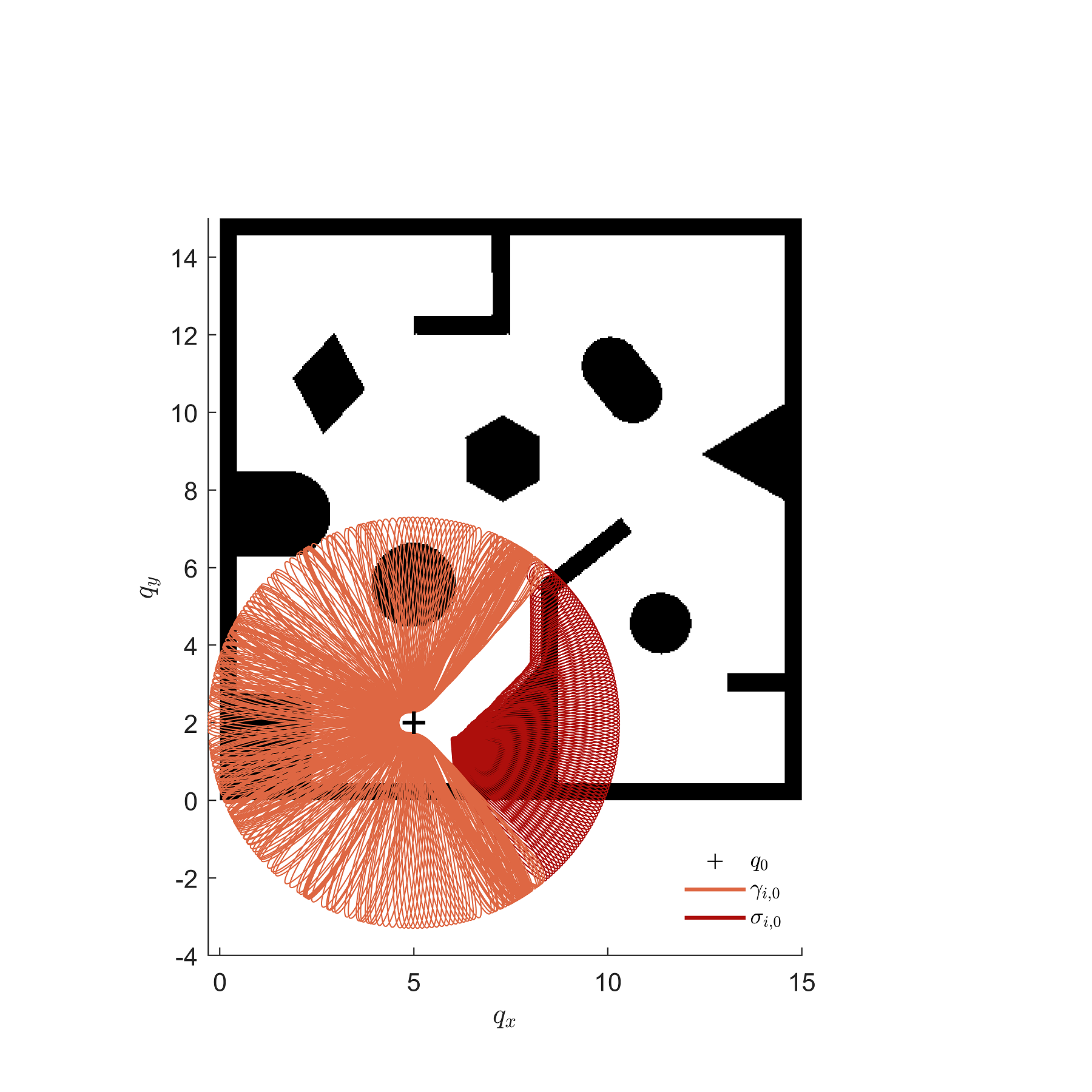}}
\caption{Map of \textit{a priori} unknown environment.}\label{fig:ex1_FOV_safeset}
\end{figure} 
\begin{figure}[t!]
\center{\includegraphics[width=0.44\textwidth,clip=true,trim= 0.25in 0.35in 1.1in 1.1in] {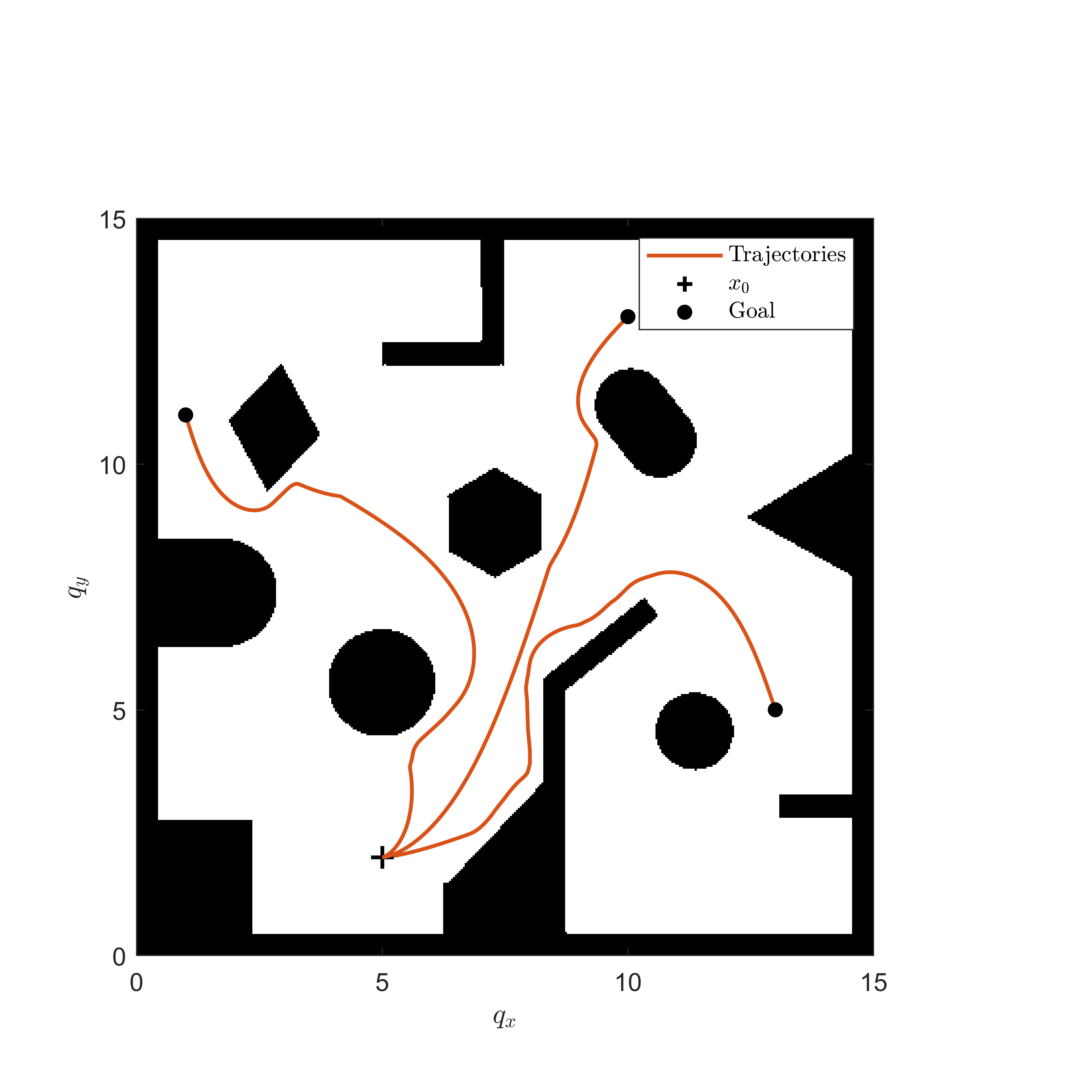}}
\caption{Three closed-loop trajectories with $100^{\circ}$ LiDAR.}\label{fig:ex1_FOV_map}
\end{figure} 
\begin{figure}[t!]
\center{\includegraphics[width=0.44\textwidth,clip=true,trim= 0.3in 0.27in 1.1in 0.5in] {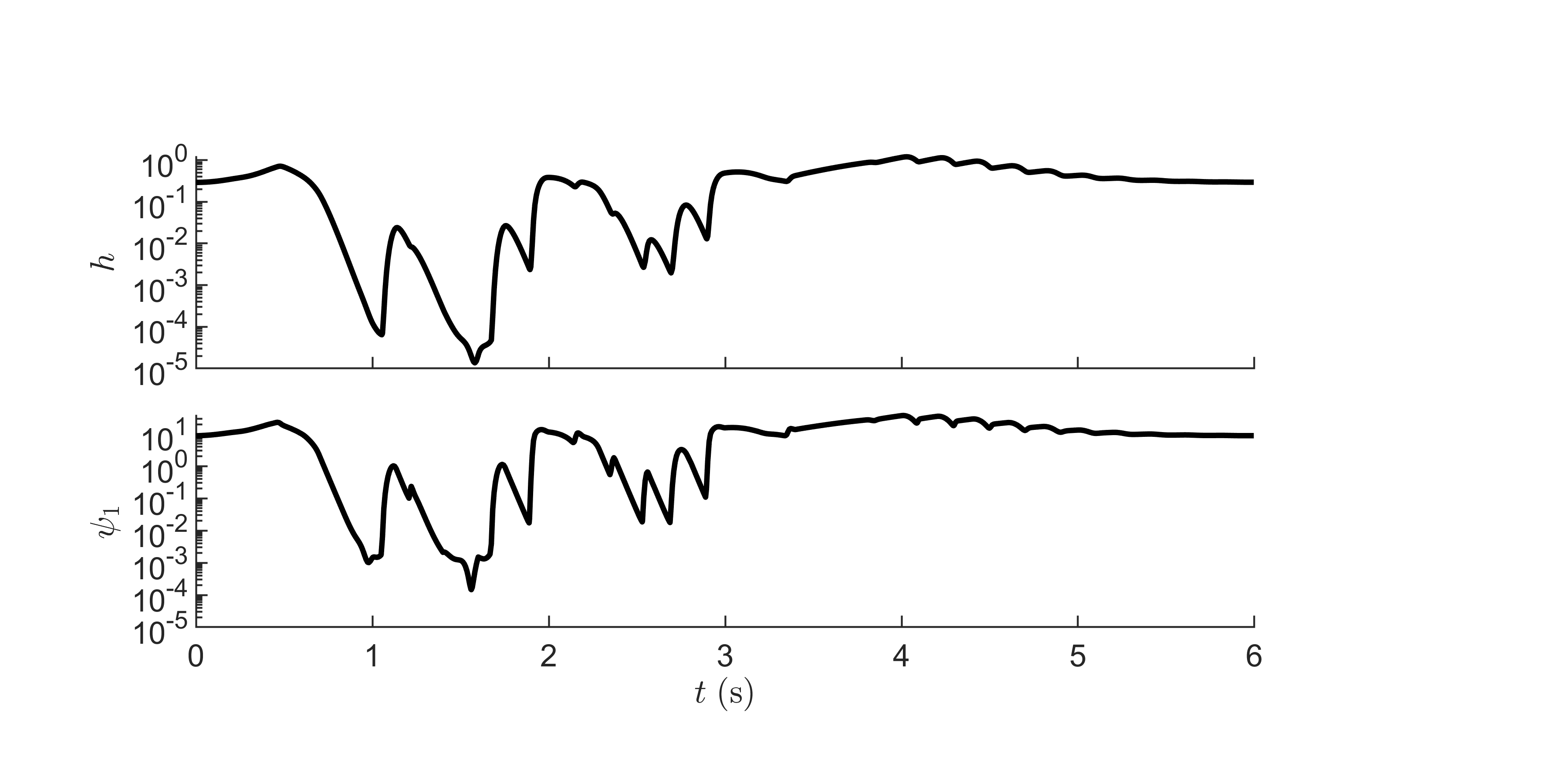}}
\caption{$h$ and $\psi_1$ for $q_\rmg = [\,13\quad5\,]^\rmT$ .}\label{fig:ex1_FOV_h}
\end{figure} 
\begin{figure}[t!]
\center{\includegraphics[width=0.44\textwidth,clip=true,trim= 0.4in 0.47in 1.1in 0.7in] {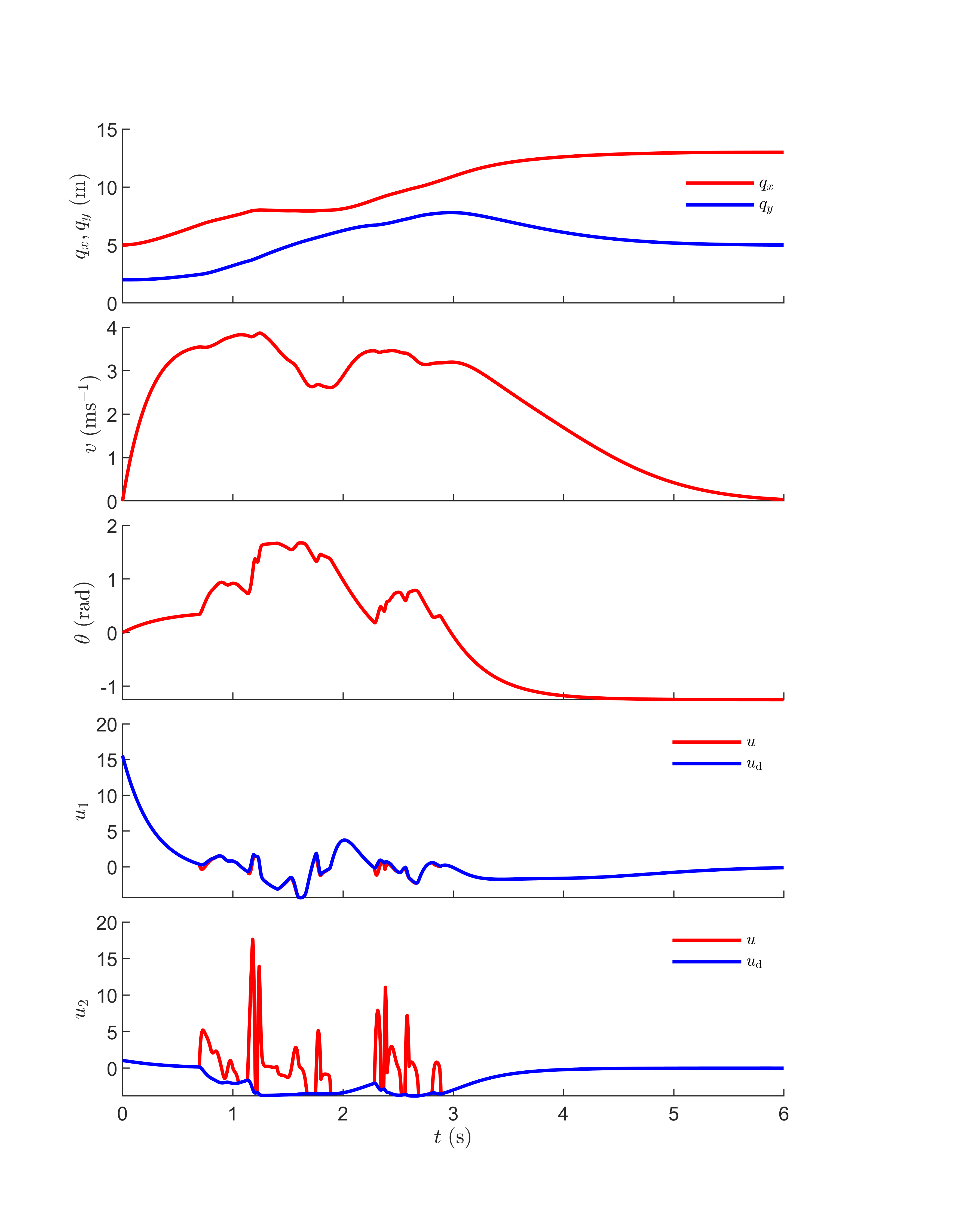}}
\caption{$q_x$, $q_y$, $v$, $\theta$, $u_\rmd$, and $u$ for $q_\rmg = [\,13\quad5\,]^\rmT$.}\label{fig:ex1_FOV_state}
\end{figure} 
\bibliographystyle{ieeetr}
\bibliography{Softmax_LiDAR.bib}

\end{document}